\documentclass{article}
\usepackage[utf8]{inputenc}
\usepackage[utf8]{inputenc}
\usepackage{amsmath}
\usepackage{amsthm}
\usepackage{amsfonts}
\usepackage{amssymb}
\usepackage{amsmath}
\usepackage{xcolor,comment}
\usepackage{cancel}
\newtheorem{theorem}{Theorem}

\newtheorem{lemma}{Lemma}

\newcommand{\mkq}[1]{{{\color{black}#1}}}
\def\f2{{\mathbb F}_{2}}

\def\fm{{\mathbb F}_{2^m}}

\newcommand{\tr}{\mathrm{tr}}
\newcommand{\cF}{\mathcal{F}}
\usepackage{enumerate}
\usepackage[overload]{empheq}

\title{Two new families of bivariate APN functions}
\author{Marco Calderini\thanks{Dipartimento di Matematica, Universit\`a degli studi di Trento, Italy. marco.calderini@unitn.it} \and Kangquan Li\thanks{College of Liberal Arts and Sciences, National University of Defense Technology, Changsha, 410073, China, likangquan11@nudt.edu.cn.} \and Irene Villa\thanks{Dipartimento di Matematica, Universit\`a degli studi di Trento, Italy. irene.villa@unitn.it} }
\date{}

\begin{document}

\maketitle
\begin{abstract}
    In this work, we present two new families of quadratic APN functions. The first one (F1) is constructed  via biprojective polynomials. This family includes one of the two APN families  introduced by G\"olo\v{g}lu in 2022. Then, following a similar approach as in Li \emph{et al.} (2022), we give another family (F2) obtained by adding  certain terms to F1. As a byproduct, this second family includes one of the two families introduced by  Li \emph{et al.} (2022).    
    Moreover, we show that for $n=12$, from our constructions, we can obtain APN functions  which are CCZ-inequivalent to any other known APN function over $\mathbb{F}_{2^{12}}$.
\end{abstract}

\section{Introduction}

Given two positive integers $n$ and $m$, set $\mathbb F_{2^n}$ and $\mathbb F_{2^m}$ be the finite fields with $2^n$ and $2^m$ elements respectively.
A function from $\mathbb{F}_{2^n}$ to $\mathbb{F}_{2^m}$ is called vectorial Boolean function or $(n,m)$-function. 
Vectorial Boolean functions play an important role in many different areas of mathematics, computer science and engineering. 
In the field of cryptography, and particularly in the design of block ciphers, $(n,m)$-functions are of critical importance, as these are usually the only nonlinear components, and as such, the security of the encryption directly depends on the properties of the $(n,m)$-functions. 

One of the most efficient attacks that can be employed against block ciphers, is the differential cryptanalysis \cite{bihSha1991}. This attack is based on the study of how differences in input can affect the resultant difference at the output. 
The resistance to differential cryptanalysis for a function $F$ from $\mathbb{F}_{2^n}$ to $\mathbb{F}_{2^m}$, used as an S-box in a cipher, is high when the value 
$$
\delta_F = \max_{a\in\mathbb{F}_{2^n}^*,b\in\mathbb{F}_{2^m}} |\{x\in\mathbb{F}_{2^n}\,:\,F(x+a)+F(x)=b\}|.
$$
is small ($\mathbb{F}_{2^n}^*=\mathbb{F}_{2^n}\setminus\{0\}$). When $n = m$, the differential uniformity of any $(n, n)$-function is at least $2$. Functions meeting this bound are called almost perfect nonlinear (APN). 

Discovering new examples and constructions of APN functions is thus a matter of significant practical importance in cryptography. 
Moreover, APN functions are also interesting from a theoretical point of view since there are several known connections between APN functions and other combinatorial and geometrical objects like semi-biplanes, difference sets, distance-regular graphs or dimensional dual hyperovals (see \cite{CGT16,CH99,DO68,DE14}).

The APNness of functions is preserved by some equivalence relations. Among these relations we have the so-called CCZ- and EA-equivalences, and it is important when several functions are considered, to determine whether they correspond to each other by such equivalences. 
CCZ-equivalence is the most general known equivalence relation preserving the APN property \cite{ccz}. 

To date, only six infinite families of APN monomials and more or less 15 (depending on how we count) infinite families of quadratic APN polynomials are known  (for a list of known APN families see \cite{CBC21}).

Recently, G\"olo\v{g}lu introduced the framework of $(q,r)$-biprojective functions \cite{GOLOGLU}. He showed that roughly half of the families of quadratic APN functions defined over an extension of even degree fall into this framework. Moreover, G\"olo\v{g}lu  introduced two new infinite families of APN functions defined over $n=2m$ with $\gcd(3,m)=1$.
Another family, coming from the framework of biprojective functions, has been determined also in \cite{golkol} for the case $m\equiv 2\mod 4$.

One of the families given in \cite{GOLOGLU}, which is reported here in Theorem \ref{th:gol} and corresponds to the construction $\cF_1$, includes the $\mathcal{BHK}$ family of quadrinomial APN functions introduced in \cite{BHK}. More precisely, $\cF_1$ includes the cases for which $\mathcal{BHK}$ family can produce functions inequivalent to APN mappings from other known families.
Moreover, G\"olo\v{g}lu showed that from family $\cF_1$ we can obtain more functions than those coming from the quadrinomials given in \cite{BHK}. Both the constructions require that the functions are defined over $\mathbb{F}_{2^n}$ with $n=2m$, and $3\nmid m$.

Later, Li \emph{et al.} \cite{LZLQ}, inspired by Dillon's method \cite{dillon}, obtained a new family of APN functions by adding some terms of the form $\sum_i(a_ix^{2^i}y^{2^i},b_ix^{2^i}y^{2^i})$ to family $\mathcal{F}_1$. Also in this case, it is needed the restriction $3\nmid m$.

In this work, we introduce two new families of quadratic APN functions. The first one is based on biprojective polynomials. This family can be defined for any value of $m$ and we have also that G\"olo\v{g}lu's family $\cF_1$ is, actually, contained in our construction. 
The second family is obtained by using a similar approach as in \cite{LZLQ}, that is adding certain terms of type $\sum_i(a_ix^{2^i}y^{2^i},b_ix^{2^i}y^{2^i})$ to our first family. This family includes the one obtained in \cite{LZLQ}.

Moreover, from our constructions we can obtain new APN functions, that is, functions which are CCZ-inequivalent to those belonging to any other known family.


\section{Preliminaries}

Let $n$ be a positive integer. We denote by $\mathbb{F}_{2^n}$ the finite field with $2^n$ elements, and by $\mathbb{F}_{2^n}^*$ the set of its non-zero elements, i.e. its multiplicative group.


Any $(n,n)$-function $F : \mathbb{F}_{2^n} \rightarrow \mathbb{F}_{2^n}$ can be expressed as a polynomial of the form $$F(x) = \sum_{i = 0}^{2^n-1} a_i x^i,$$ for $a_i \in \mathbb{F}_{2^n}$. This is called the \textit{univariate representation} of $F$, and it is unique.

The \textit{algebraic degree} of $F$, denoted by $\deg(F)$, is the largest binary weight of an exponent $i$ with $a_i \ne 0$ in the univariate representation of $F$, where the \textit{binary weight} of an integer is the number of ones in its binary representation. Functions of algebraic degree $1$, resp, $2$ are called \textit{affine}, resp. \textit{quadratic}. An affine function $F$ satisfying $F(0) = 0$ is called \textit{linear}.

For $m \mid n$, we denote by $\mathrm{tr}^n_m : \mathbb{F}_{2^n} \rightarrow \mathbb{F}_{2^m}$ the trace function $\mathrm{tr}^n_m(x) = \sum_{i = 0}^{n/m - 1} x^{2^{mi}}$. If $m=1$ we  denote  $\mathrm{tr}^n_1$ by $\mathrm{tr}_n$.

The \textit{Walsh transform} of $F : \mathbb{F}_{2^n} \rightarrow \mathbb{F}_{2^n}$ is defined as $$W_F(a,b) = \sum_{x \in \mathbb{F}_{2^n}} (-1)^{ \tr_n(a x+bF(x))}$$ for $a,b \in \mathbb{F}_{2^n}$. The Boolean function $ F_b(x)= \tr_n(b F(x))$, for $b\in \mathbb{F}_{2^n}^*$, is called a {\em component function} of $F$. 
A component function $F_b$ satisfying $|W_F(a,b)|=2^{n/2}$ for any $a\in\mathbb{F}_{2^n}$ is called \textit{bent}. Bent functions are defined only for $n$ even.

Two $(n,n)$-functions $F$ and $G$ are said to be EA-equivalent if $G = A_1 \circ F \circ A_2 + A$ for affine $A_1, A_2, A : \mathbb{F}_{2^n} \rightarrow \mathbb{F}_{2^n}$ with $A_1, A_2$ bijective. 

We say that two $(n,n)$-functions $F$ and $G$ are CCZ-equivalent if there is an affine permutation $\mathcal{L}$ of $\mathbb{F}_{2^n}^2$ which maps the graph $\Gamma_F = \{ (x,F(x)) : x \in \mathbb{F}_{2^n} \}$ of $F$ to the graph $\Gamma_G$ of $G$. EA-equivalence is a special case of CCZ-equivalence, and CCZ-equivalence is more general than EA-equivalence \cite{BCP}. However, when we restrict to quadratic APN functions, we have that EA-equivalence coincides with CCZ-equivalence \cite{ccz-ea}. 

To prove that a family of APN functions is new, it is necessary to show that we can obtain instances of APN functions which are CCZ-inequivalent to those of the currently known APN families. A common approach is to compare some invariants, i.e. properties that are preserved under certain equivalence relations.

For a given map $F : \mathbb{F}_{2^n} \rightarrow \mathbb{F}_{2^n}$ we define the set
$$
\mathrm{NB}_F=\{b\in\mathbb{F}_{2^n}\,:\,W_F(a,b)=0 \text{ for some } a\in\mathbb{F}_{2^n}\},
$$
and the sequence
$$
N_F=[n_i(\mathrm{NB}_F)\,:\, 1\le i\le n ],
$$
where $n_i(S)$ is the number of $\mathbb{F}_2$-vector spaces of dimension $i$ in $S$.

For a quadratic APN function $F$, we have that $\mathrm{NB}_F$ is the set of non-bent components, that is,
$$
\mathrm{NB}_F=\{b\in\mathbb{F}_{2^n}\,:\,|W_F(0,b)|\ne 2^{n/2}\}.
$$

The sequence $N_F$ has been shown to be an EA-invariant \cite{seta18-1,golpav}.

When $n=2m$, we can identify $\mathbb{F}_{2^n}$ with $\mathbb{F}_{2^m}\times \mathbb{F}_{2^m}$. In this case a function
$F : \mathbb{F}_{2^n} \rightarrow \mathbb{F}_{2^n}$
can be represented as a (univariate) polynomial in  $\mathbb{F}_{2^n}[X]$, or a (bivariate) polynomial in $\mathbb{F}_{2^m}[x,y]\times \mathbb{F}_{2^m}[x,y]$. 

The idea of considering a bivariate form for constructing APN functions was firstly considered in \cite{CARLET1}. Here, Carlet considered functions $F$ defined over $\mathbb{F}_{2^{2m}}$ given by $F(x,y)=(f(x,y),g(x,y))$, with $f$ the Maiorana-McFarland bent function $xy$ from $\mathbb{F}_{2^m}\times\mathbb{F}_{2^m}$ to $\mathbb{F}_{2^m}$. 
Using this framework, Carlet introduced a family of APN functions that, as shown in \cite{CBC21}, coincides with the family of hexanomial APN functions given in \cite{BudCar08}.
Other two families, which consider  $f(x,y)=xy$, have been introduced later in \cite{tan19,pott}.

The aforementioned construction has been extended, recently, by G\"olo\v{g}lu in \cite{GOLOGLU}, where the author introduced the framework of biprojective polynomials.

In particular, let $q=2^k$, we call a polynomial of type $ax^{q+1} + bx^q + cx+d\in\mathbb{F}_{2^m}[x]$ \textit{projective}, and $ax^{q+1} + bx^q y + cxy^q + dy^{q+1}\in\mathbb{F}_{2^m}[x,y]$ a \textit{bivariate projective} (or biprojective) polynomial.

G\"olo\v{g}lu, in his work, introduced the following two families of APN functions using biprojective polynomials.

\begin{theorem}[\cite{GOLOGLU}]\label{th:gol}
The following functions are APN on $\mathbb{F}_{2^{2m}}$:
\begin{itemize}
\item[$\cF_1$:] If $\gcd(3k, m) = 1$, let $q=2^k$
$$
F(x,y)=(x^{q+1} +xy^{q}+y^{q+1},x^{q^2+1}+x^{q^2}y+y^{q^2+1});
$$
\item[$\cF_2$:] If $\gcd(3k, m) = 1$, $m$ odd, let $q=2^k$
$$
F(x,y)=(x^{q+1} +xy^{q}+y^{q+1},x^{q^3}y+xy^{q^3}).
$$
\end{itemize}
\end{theorem}

Later in \cite{LZLQ}, the authors considered family $\mathcal{F}_1$, with $k=1$, and showed that by adding the term $(xy,xy+x^2y^2)$ to it, we can obtain another family of APN functions. The resulting family is the following.

\begin{theorem}[\cite{LZLQ}]\label{th:lzlq}
Let $m>0$ such that $\gcd(3,m)=1$. Then, the function 
$$
F(x,y)=(x^{3}+xy +xy^{2}+y^{3},x^{5}+xy+x^2y^2+x^{4}y+y^{5})
$$
is APN.
\end{theorem}

In the following section we introduce two new families constructed using biprojective polynomials and, thereafter, by adding certain quadratic terms to it. These two families can be defined for any value $m$. Our constructions include family $\cF_1$ and the functions in Theorem \ref{th:lzlq}.

\section{A new family of  biprojective APN functions over $\mathbb{F}_{2^{2m}}$}
In the following lemma we give some properties on projective polynomials that we use along the paper. Some of these facts can be found also in \cite{GOLOGLU}, we report the proof for completness.

\begin{lemma}\label{lm:prop}
Let $k,m>0$ with $\gcd(k,m)=1$, $q=2^k$ and $\alpha\in \mathbb{F}_{2^m}$. Let us denote by $\phi_{q,\alpha}(x)=x^{q+1}+x+\alpha$.
\begin{itemize}
    \item[(i)] $\phi_{q,\alpha}(x)$ has no root over $\mathbb{F}_{2^m}$ if and only if $\phi_{q,\alpha^2}(x)$  has no root over $\mathbb{F}_{2^m}$. The same is verified also for $\alpha^q x^{q+1}+x+\alpha$.
    \item[(ii)] $ax^{q+1}+bx+c$ has no root over $\mathbb{F}_{2^m}$ if and only if $ax^{q^2}+bx^q+cx$ is a permutation.
    \item[(iii)] Let $a,b,c,d\in \mathbb{F}_{2^m}$. Suppose that $\phi_{q,\alpha}(x)$ has no root over $\mathbb{F}_{2^m}$. Then, if $ad-bc\ne 0$, 
    $$
    \phi(x)=(ax+b)^{q+1}\phi_{q,\alpha}\left(\frac{cx+d}{ax+b}\right)
    $$
    has no root over $\mathbb{F}_{2^m}$.
\end{itemize}
\end{lemma}
\begin{proof}
\noindent {\bf (i):}

We have immediately that 
$$
\phi_{q,\alpha}(x^{2^{m-1}})^2=\phi_{q,\alpha^2}(x).
$$
So, $\phi_{q,\alpha}(x)$ admits roots in $\mathbb{F}_{2^m}$ if and only if $\phi_{q,\alpha^2}(x)$ does.

Similar, we have $1/\alpha\ \phi_{q,\alpha^2}(\alpha x)=\alpha^q x^{q+1}+x+\alpha$.

\noindent {\bf (ii):}

Since $\gcd(k,m)=1$, we have that $x^{q-1}$ permutes $\mathbb{F}_{2^m}$. Let $p(x)=ax^{q+1}+bx+c$, then
$$
xp(x^{q-1})=ax^{q^2}+bx^q+cx.
$$
So there is a 1-to-1 correspondence between the roots of $p$ and the nonzero elements in the kernel of $ax^{q^2}+bx^q+cx$.

\noindent {\bf (iii):} 
We can note that $\phi$ is a polynomial,
$$
\phi(x)=(cx+d)^{q+1}+(cx+d)(ax+b)^{q}+\alpha(ax+b)^{q+1}.
$$
If $ax+b\ne 0$, then $\phi$ is the product of two nonzero elements. 

For $ax+b= 0$, that is $x=\frac{b}{a}$, we have that $\phi(x)=(cx+d)^{q+1}=\left(c\frac{b}{a}+d\right)^{q+1}$, which is different from zero.
\end{proof}

{We present now the main result of this section, which extends the mentioned family $\cF_1$.}

\begin{theorem}[Family F1]\label{th:main}
Let $k,m>0$, with $\gcd(k,m)=1$. Let $q=2^k$, and let $F:{\Bbb F}_{2^m}\times{\Bbb F}_{2^m}\rightarrow{\Bbb F}_{2^m}\times{\Bbb F}_{2^m}$ given by
$$F(x,y)=(f(x,y),g(x,y))=(x^{q+1}+xy^q+\alpha y^{q+1}, x^{q^2+1}+\alpha x^{q^2}y+(1+\alpha)^qxy^{q^2}+\alpha y^{q^2+1}),$$
with $\alpha\in\mathbb{F}_{2^m}$ such that $ x^{q+1}+x+\alpha$ has no root in $\mathbb{F}_{2^m}$. Then, $F$ is APN.
\end{theorem}
\begin{proof}
First of all, notice that from the restriction on $\alpha$ we easily deduce that $\alpha\ne0$.
Since $F$ is quadratic, we need to check the number of solutions of the equation
\begin{equation}\label{eq:der}
F(x+v,y+w)+F(x,y)+F(v,w)=0,
\end{equation}
for any $(v,w)\in{\Bbb F}_{2^m}\times{\Bbb F}_{2^m}\setminus\{(0,0)\}$.
Let $D_f^{[v,w]}(x,y)=f(x+v,y+w)+f(x,y)+f(v,w)$ and $D_g^{[v,w]}=g(x+v,y+w)+g(x,y)+g(v,w)$. As in \cite{GOLOGLU}, we can define 
\begin{align*}
  E_f^0(x,y)=&\frac{D_f^{[0,w]}(wx,wy)}{w^{q+1}}=x+\alpha (y^q+y),  \\
  E_g^0(x,y)=&\frac{D_g^{[0,w]}(wx,wy)}{w^{q^2+1}}=\alpha x^{q^2}+(\alpha+1)^qx+\alpha (y^{q^2}+y),\\
  E_f^\infty(x,y)=&\frac{D_f^{[v,0]}(vx,vy)}{v^{q+1}}=x^q+x+y^q,\\
  E_g^\infty(x,y)=&\frac{D_g^{[v,0]}(vx,vy)}{v^{q^2+1}}= x^{q^2}+x+(\alpha+1)^qy^{q^2}+\alpha y,
\end{align*}
and for the case $vw\ne 0$, denoting by $u=v/w$,
\begin{align*}
   E_f^u(x,y)=&\frac{D_f^{[v,w]}(vx,wy+wx)}{w^{q+1}}=f(u,1)(x^q+x)+(u+\alpha) y^q+\alpha y, \\
   E_g^u(x,y)=&\frac{D_g^{[v,w]}(vx,wy+wx)}{w^{q^2+1}}=g(u,1)(x^{q^2}+x)+((\alpha+1)^qu+\alpha) y^{q^2}+\alpha(u+1)^{q^2} y.
\end{align*}

As shown in \cite{GOLOGLU}, checking the number of solutions of \eqref{eq:der} is equivalent to checking the number of solutions of $E_f^0(x,y)=E_g^0(x,y)=0$, $E_f^\infty(x,y)=E_g^\infty(x,y)=0$, and $E_f^u(x,y)=E_g^u(x,y)=0$, depending on the values of $v$ and $w$.

Let us consider the first case, that is, $E_f^0(x,y)=E_g^0(x,y)=0$. In this case we have
$$
\frac{x}{\alpha}= y^q+y,
$$
$$
\alpha x^{q^2}+(\alpha+1)^qx+\alpha (y^{q^2}+y)=\alpha x^{q^2}+(\alpha+1)^qx+\alpha [(y^{q}+y)^q+(y^{q}+y)]=0.
$$
Therefore,
$$
\begin{aligned}
0=&\alpha x^{q^2}+(\alpha+1)^qx+\alpha (y^{q^2}+y)\\
=&\alpha x^{q^2}+\frac{x^q}{\alpha^{q-1}}+\alpha^qx.\\
\end{aligned}
$$

Substituting $x\mapsto\alpha x$ we get
$$
\begin{aligned}
0=&\alpha^{q^2+1} x^{q^2}+\alpha{x^q}+\alpha^{q+1}x\\
=&\alpha^{q^2} x^{q^2}+{x^q}+\alpha^{q}x.
\end{aligned}
$$

From Lemma \ref{lm:prop}, we have that $\alpha^{q^2} x^{q^2}+{x^q}+\alpha^{q}x$ permutes $\mathbb{F}_{2^m}$. Then, $x=0$, implying $ y^q+y=0$, and thus $y\in\mathbb{F}_2$. So, in this case we have two solutions in $\{0\}\times\mathbb{F}_2$.

Consider now the case $E_f^\infty(x,y)=E_g^\infty(x,y)=0$. Then,
$$
y^q=x^q+x
$$
and
$$
0=x^{q^2}+x+(\alpha+1)^qy^{q^2}+\alpha y.
$$
Therefore,
$$
\begin{aligned}
0=&y^{q^2}+y^q+(\alpha+1)^qy^{q^2}+\alpha y\\
=&\alpha^q y^{q^2}+y^q+\alpha y.\\
\end{aligned}
$$
Applying Lemma \ref{lm:prop} we have that $\alpha^q y^{q^2}+y^q+\alpha y$ is a permutation. Therefore, we have two solutions $(x,y)\in\mathbb{F}_2\times\{0\}$.

Now, we consider the general case $vw\ne 0$ and thus the case $E_f^u(x,y)=E_g^u(x,y)=0$, with $u\ne0$.
In the following, we will denote by $f_u=f(u,1)=u^{q+1}+u+\alpha$ and by $g_u=g(u,1)=u^{q^2+1}+\alpha u^{q^2}+u+\alpha^qu+\alpha$.
From the restriction on $\alpha$ we have $f_u\ne0$ for any $u$.

We report equations $E_f^u(x,y)=0$ and  $E_g^u(x,y)=0$.
\begin{align*}
   f_u(x^q+x)=&{(u+\alpha)}y^q+{\alpha} y\\
   =&\left(f_u+{u^{q+1}}\right)y^q+\left(f_u+{u(u+1)^q}\right)y\\
    g_u(x^{q^2}+x)=&{((1+\alpha)^qu+\alpha)}y^{q^2}+{\alpha(u^{q^2}+1)}y\\
    =&\left(g_u+{u^{q^2}(u+\alpha)}\right)y^{q^2}+\left(g_u+{u(u^{q^2}+1+\alpha^q)}\right)y
\end{align*}
We perform  a change of variables. First we substitute $x\mapsto x+y$, then we substitute $y\mapsto u^{-1}y$. 
Therefore, we obtain
\begin{align}
  f_u (x^q+x)
   =&uy^q+{(u+1)^q}y,\label{eq:A} \\
    g_u(x^{q^2}+x)=&({u+\alpha})y^{q^2}+({u^{q^2}+1+\alpha^q})y. \label{eq:B}
\end{align}
Since $x^{q^2}+x=(x^q+x)^q+(x^q+x)$, then 
$f_u^{q+1}g_u(x^{q^2}+x)=f_ug_u\left(f_u(x^q+x)\right)^q+f_u^qg_u\left(f_u(x^q+x)\right)$.
Thus,  using  Eqs.\ \eqref{eq:A} and \eqref{eq:B}, we obtain the  following:
\begin{align*}
0=&f_u^{q+1}\left(({u+\alpha})y^{q^2}+({u^{q^2}+1+\alpha^q})y\right)
+f_ug_u\left(uy^q+{(u+1)^q}y\right)^q+f_u^qg_u\left(uy^q+{(u+1)^q}y\right)\\
=&f_uy^{q^2}\left(f_u^q(u+\alpha)+g_uu^q\right)
+g_uy^q\left(f_u(u+1)^{q^2}+f_u^qu\right)\\
&+f_u^qy\left(f_u(u^{q^2}+1+\alpha^q)+g_u(u+1)^q\right).
\end{align*}
By explicitly computing $f_u$ and $g_u$, we obtain the following equalities.
\begin{align*}
    & f_u^q(u+\alpha)+g_uu^q=\\
    &\quad=
    (u^{q^2+q}+u^q+\alpha^q)(u+\alpha)+(u^{q^2+1}+u^{q^2}\alpha+u\alpha^q+u+\alpha)u^q\\
    &\quad=u^{q^2+q+1}+u^{q+1}+\alpha^qu+\alpha u^{q^2+q}+\alpha u^q+\alpha^{q+1}+ \\
   & \quad\quad+u^{q^2+q+1}+\alpha u^{q^2+q}+\alpha^q u^{q+1}+u^{q+1}+\alpha u^q\\
    &\quad=\alpha^q(u^{q+1}+u+\alpha)=\alpha^qf_u,\\
    &f_u(u+1)^{q^2}+f_u^qu=\\
    &\quad=(u^{q+1}+u+\alpha)(u^{q^2}+1)+(u^{q^2+q}+u^q+\alpha^q)u\\
    &\quad=u^{q^2+q+1}+u^{q^2+1}+\alpha u^{q^2}+u^{q+1}+u+\alpha+ u^{q^2+q+1}+u^{q+1}+\alpha^qu\\
    &\quad= u^{q^2+1}+u^{q^2}\alpha+u\alpha^q+u+\alpha= g_u,\\
    & f_u(u^{q^2}+1+\alpha^q)+g_u(u+1)^q=\\
     &\quad= (u^{q+1}+u+\alpha)(u^{q^2}+1+\alpha^q)+(u^{q^2+1}+u^{q^2}\alpha+u\alpha^q+u+\alpha)(u^q+1)\\
     &\quad= u^{q^2+q+1}+u^{q^2+1}+\alpha u^{q^2}+u^{q+1}+u+\alpha+\alpha^qu^{q+1}+\alpha^qu+\alpha^{q+1}+\\
     &\quad\quad+u^{q^2+q+1}+\alpha u^{q^2+q}+\alpha^qu^{q+1}+u^{q+1}+\alpha u^q+u^{q^2+1}+\alpha u^{q^2}+\alpha^qu+u+\alpha\\
     &\quad= \alpha(u^{q^2+q}+u^q+\alpha^q)=
     \alpha f_u^q.
\end{align*}

    
Therefore, we obtain  the  equation
\begin{equation}\label{eq:a}
        \alpha^qf_u^2y^{q^2}+g_u^2y^q+\alpha f_u^{2q}y=0.
\end{equation}

 We will show that 
$$
L(y)=\alpha^qf_u^2y^{q^2}+g_u^2y^q+\alpha f_u^{2q}y
$$
is a permutation. This is equivalent (see Lemma \ref{lm:prop}) to show that we have no root in $\mathbb{F}_{2^m}$ for 
$$
\phi(y)=\alpha^qf_u^2y^{q+1}+g_u^2y+\alpha f_u^{2q}.
$$

Now, we have
$$
\phi'(y)=\phi(y^2)^{2^{m-1}}=\alpha^{q2^{m-1}}f_uy^{q+1}+g_uy+\alpha^{2^{m-1}} f_u^{q},
$$
substituting $y\mapsto y/\alpha^{2^{m-1}}$ and multiplying by $\alpha^{2^{m-1}}$ we get
$$
\phi''(y)=f_uy^{q+1}+g_uy+\alpha f_u^{q}.
$$

By computation, we have that 
$$
\begin{aligned}
\phi''(y)=&(y+(u+1)^q)^{q+1}\phi_{q,\alpha}\left(\frac{uy+\alpha}{y+(u+1)^q}\right).\\
\end{aligned}
$$
Since 
$$
\left|
\begin{array}{cc}
  1   & (u+1)^q \\
 u    & \alpha
\end{array}\right|=f_u\ne 0,
$$
from Lemma \ref{lm:prop}, we obtain that $\phi''(y)$ has no root in $\mathbb{F}_{2^m}$.

Therefore, the only solution of \eqref{eq:a} is $y=0$ and from Eq.\ \eqref{eq:A} we obtain two solutions $(x,y)\in\mathbb{F}_2\times\{0\}$.
\end{proof}

The existence of our functions is based on the existence of elements $\alpha$ for which the projective polynomial $\phi_{q,\alpha}(x)$ has no roots.

Projective polynomials and their roots have been studied in several works. Bluher's paper \cite{BLU} has become a standard reference on this topic. In particular, from the results in  \cite{BLU}, we have that for any $m$ and $k$, with $\gcd(k,m)=1$, there exists $\alpha\in \mathbb{F}_{2^m}$ such that $x^{2^k+1}+x+\alpha$ has no root in $\mathbb{F}_{2^m}$ (see also \cite{HK} for existence of these elements).

Family F1 defined in Theorem \ref{th:main} contains family $\mathcal{F}_1$ given in Theorem \ref{th:gol}. Indeed, the projective polynomial $x^{2^k+1}+x+1$ has no root over $\mathbb{F}_{2^m}$ if and only if $\gcd(3,m)=1$. Therefore, our family generalizes family $\cF_1$, since our construction does not require any restriction on $m$.

\section{A new family of APN function from Dillon's method}

In this section we  show that, as for the case of family $\mathcal{F}_1$, adding certain quadratic terms to our functions, we can obtain another family of APN functions. { Before that, we first give two useful lemmas. }

{	\begin{lemma}[\cite{williams1975note}]	\label{cubic}
	
	Let  $a,b\in\fm$, $b\neq0$ and define 
	$$f(z)=z^3+az+b,\quad h(t)=t^2+bt+a^3.$$ 
	Let $t_1, t_2$ be two solutions of $h(t)=0$ in $\mathbb{F}_{2^{2m}}$. Then: 
	\begin{enumerate}[(1)]
		\item $f$ has three zeros in $\fm$ if and only if $\tr_m\left(\frac{a^3}{b^2}\right)=\tr_m(1)$ and $t_1, t_2$ are cubes in $\fm$ (resp. $\mathbb{F}_{2^{2m}}$) when $m$ is even (resp. odd);
		\item $f$ has exactly one zero in $\fm$ if and only if $\tr_m\left(\frac{a^3}{b^2}\right)\neq\tr_m(1)$;
		\item $f$ has no zeros in $\fm$ if and only if $\tr_m\left(\frac{a^3}{b^2}\right)=\tr_m(1)$ and $t_1, t_2$ are not cubes in $\fm$ (resp. $\mathbb{F}_{2^{2m}}$) when $m$ is even (resp. odd).
	\end{enumerate}
\end{lemma}} 

\mkq{\begin{lemma}
		\label{lemma1}
	Let $\alpha\in \fm$ such that  $ x^{3}+x+\alpha$ has no root in $\mathbb{F}_{2^m}$. Then 
	\begin{enumerate}[(1)]
		\item 	$\alpha\neq0$,  $\tr_m\left( \frac{1}{\alpha^2} \right) = \tr_m(1)$ and there exists some $\beta\in\fm$ such that $\alpha(\beta + \omega)$ is not a cube in  $\fm$ (resp. $\mathbb{F}_{2^{2m}}$) when $m$ is even (resp. odd) and such that $\frac{1}{\alpha^2} = \beta + \beta^2+1$, where $\omega\in\mathbb{F}_{2^2}\backslash\f2$;
		\item  the equation
		\begin{equation}
			\label{lemma-eq1} a^3+\alpha b a^2 + (\alpha^2 b^2 + b^2 +1) a + \alpha^3b^3+\alpha b^2 + \alpha =0
		\end{equation}  has only one solution $(a,b)=(\alpha,1)$ in $\fm\times \fm$. 
	\end{enumerate}
\end{lemma}
\begin{proof}
(1) First, it is trivial that $\alpha\neq 0$. Then since  $ x^{3}+x+\alpha$ has no root in $\mathbb{F}_{2^m}$,    by Lemma \ref{cubic}, we have  $\tr_m\left( \frac{1}{\alpha^2} \right) = \tr_m(1)$ and $t_1$ is not cubic in $\fm$ (resp. $\mathbb{F}_{2^{2m}}$) when $m$ is even (resp. odd), where $t_1$ is a solution of the equation $t^2+\alpha t+1 = 0$. Finally, it is easy to obtain that $\alpha(\beta + \omega)$ is a solution of $t^2+\alpha t+1 = 0$ and thus the statement holds.

(2) Let $a_1=a+\alpha b$. Then by simplifying, Eq.\ \eqref{lemma-eq1} becomes 
	$$a_1^3+ (b^2+1)a_1 + \alpha (b+1)^3 = 0.$$
	If $b=1$, then $a_1=0$ and thus $a=\alpha$. If $b\neq 1$, then $$\left(\frac{a_1}{b+1}\right)^3 + \frac{a_1}{b+1} + \alpha =0 $$
	has no solution in $\fm$ thanks to the condition that $ x^{3}+x+\alpha$ has no root in $\mathbb{F}_{2^m}$. Therefore, Eq.\ \eqref{lemma-eq1} has only one solution $(a,b)=(\alpha,1)$ in $\fm\times \fm$. 
\end{proof}
}

\mkq{{Since} the resultant of polynomials will be used in our proof, we now recall some basic facts about it. Given two polynomials $ u(x) = a_mx^m+a_{m-1}x^{m-1}+\cdots+a_0$ and $ v(x) = b_nx^n+b_{n-1}x^{n-1}+\cdots+b_0 $ 
over a {field} $K$ with degrees $m$ and $n$, respectively, their resultant $\mathrm{Res}(u,v)\in {K}$ is the determinant of the following square matrix of order $n+m$:
$$ \small \begin{pmatrix} 
	a_m & a_{m-1} &  \cdots & a_0  & 0 & & \cdots & 0 \\
	0 & a_m & a_{m-1} & \cdots & a_0 & 0 & \cdots  & 0 \\
	\vdots &  &  \ddots&  &  &  & & \vdots  \\
	0 & \cdots & 0 & a_m & a_{m-1} &  & \cdots & a_0 \\
	b_n & b_{n-1} & \cdots &  & b_0 & 0 &\cdots  & 0\\
	0 & b_{n} & b_{n-1} & \cdots &  & b_0 & \cdots  & 0\\
	\vdots &  & \ddots &  & &   &  \ddots &   \vdots   \\ 
	0   & \cdots & 0 & b_{n} & b_{n-1}&  &   \cdots &  b_0    \\
\end{pmatrix}.
$$
For a field $K$ and two polynomials $F(x,y), G(x,y) \in K[x,y]$, we use $ \mathrm{Res}_y(F,G)$ to denote the resultant  of $F$ and $G$ with respect to $y$, { which is the resultant of $F$ and $G$ when considered as  polynomials in the single variable $y$.} In this case, $ \mathrm{Res}_y(F,G)\in K[x]$ belongs to the ideal generated by $F$ and $G$. {It is well known that $F(x,y)=0$ and $G(x,y)=0$ have a common solution $(x, y)$ if and only if
	$x$ is a solution of $\mathrm{Res}_y(F,G)(x)=0$. }}

\begin{theorem}[Family F2]\label{th:main2}
Let $F:{\Bbb F}_{2^m}\times{\Bbb F}_{2^m}\rightarrow{\Bbb F}_{2^m}\times{\Bbb F}_{2^m}$ given by
$$F(x,y)=(f(x,y),g(x,y))=(x^{3}+xy+xy^2+\alpha y^{3}, x^{5}+xy+\alpha x^2y^2+\alpha x^{4}y+(1+\alpha)^2 xy^{4}+\alpha y^{5}),$$
with $\alpha\in\mathbb{F}_{2^m}$ such that $ x^{3}+x+\alpha$ has no root in $\mathbb{F}_{2^m}$. Then, $F$ is APN.
\end{theorem}

{\begin{proof}
		Since  $F$ is a quadratic function {with $F(0,0)=(0,0)$}, 
	it suffices to show that for any $(a,b)\neq(0,0) \in \fm^2$, the equation 
	\begin{equation}
		\label{Eq_APN}
		F(x+a,y+b)+F(x,y)+F(a,b)=0
	\end{equation} 
	has exactly two solutions $(x,y)\in \{(0,0), (a,b)\}$ in $\fm\times \fm.$ By {a simple calculation}, Eq.\ \eqref{Eq_APN} is equivalent to the following equation system: 
	\begin{subequations} 
		\renewcommand\theequation{\theparentequation.\arabic{equation}}  
		\label{eq_1}  
		\begin{empheq}[left={\empheqlbrace\,}]{align}
			& G(x, y) = ax^2 + (a^2+b^2+b)x  + (\alpha b + a) y^2 + (\alpha b^2+a) y = 0 \label{eq_1_1} \\ 
			& H(x, y)= (\alpha b +a)x^4 + \alpha b^2x^2 + (\alpha^2b^4 + a^4+b^4 + b) x +   \notag\\
			& (\alpha^2a + \alpha b +a ) y^4 + \alpha a^2 y^2 + (\alpha a^4+\alpha b^4+a) y  = 0.  \label{eq_1_2} 
		\end{empheq}
	\end{subequations} 

	First of all, we consider the case $(a,b)=(\alpha,1)$. In this case, Eqs.\ \eqref{eq_1} become 
\begin{subequations} 
	\renewcommand\theequation{\theparentequation.\arabic{equation}}  
	\label{eq_11}  
	\begin{empheq}[left={\empheqlbrace\,}]{align}
		& \alpha x^2 + \alpha^2 x =0 \label{eq_11_1} \\ 
		& \alpha x^2 + (\alpha^4+\alpha^2) x + \alpha^3 y^4+ \alpha^3 y^2 + \alpha^5 y = 0.  \label{eq_11_2} 
	\end{empheq}
\end{subequations} 
From Eq.\ \eqref{eq_11_1}, we know $x\in \{0, \alpha \}$. If $x=0$, plugging it into Eq.\ \eqref{eq_11_2}, we have $\alpha^3 y^4+ \alpha^3 y^2 + \alpha^5 y=0$, i.e., $y(y^3+y+\alpha^2)=0$. Thus $y=0$ since $y^3+y+\alpha^2=0$ has no solution in $\fm$ according to the condition. If $x= \alpha$, together with Eq.\ \eqref{eq_11_2}, we get $
\alpha^3 y^4+ \alpha^3 y^2 + \alpha^5 y+ \alpha^5 =0$, i.e., $(y+1)(y^3+y^2+\alpha^2) = 0$. Hence $y=1$ or $y^3+y^2+\alpha^2 = 0$. Let $Y=y+1$. Then the equation $y^3+y^2+\alpha^2 = 0$ becomes $Y^3+Y+\alpha^2=0$, which has no solution in $\fm$.  Therefore in this case, Eqs.\ \eqref{eq_1} have {only $(x,y) \in \{(0,0),(\alpha, 1)\}$ as solutions} in $\fm\times \fm$.

{In the following, we always assume that $(a,b)\neq(\alpha,1)$. Note that we can see $G, H$ be polynomials in $\f2 [x,y,\alpha, a,b]$. Then with the help of MAGMA, we obtain the resultant of $G$ and $H$ with respect to $y$ as follows } 
\begin{align}
	\mathrm{Res}_y(G,H)(x) = \left(a^3+ab^2+\alpha b^3\right)^2x(x+a) R(x,a,b)R(x+a,a,b),\label{resultant}
\end{align}
where 
\begin{align*}
	R(x,a,b) =& \alpha^2 x^3 + (a^2 + (\alpha  b + \alpha) a + b^2 +\alpha^2  b    +  1 ) x + \\
	& a^3+\alpha b a^2 + (\alpha^2 b^2 + b^2 +1) a + \alpha^3b^3+\alpha b^2 + \alpha 
\end{align*}
In the sequel we shall show that $\mathrm{Res}_y(G,H)(x)=0$ is equivalent to $x(x+a)=0$.

Firstly, we have $a^3+ab^2+\alpha b^3\neq 0$ for any $(a,b)\neq(0,0) \in \fm\times\fm$. Otherwise, for some element  $(a,b)\neq(0,0) \in  \fm\times\fm$, $a^3+ab^2+\alpha b^3 =0$. If $b=0$, then the above equation becomes $a^3=0$, which contradicts the assumption $(a,b)\neq(0,0)$. If $b\neq0$, then we have $c^3+c+ \alpha=0$, where $c=\frac{a}{b}\in\fm$, which contradicts the condition that $ x^{3}+x+\alpha$ has no root in $\mathbb{F}_{2^m}$. 

In addition, we need to show $R(x,a,b)R(x+a, a,b)\neq 0$. Note that $R(x,a,b)=0$ has the same number of solutions in $\fm$ as $R(x+a,a,b)=0$. It suffices to show that the equation $R(x,a,b)=0$ has no solution in $\fm$.  

	Now we consider the equation  $R(x,a,b)=0$, i.e.,
\begin{equation}
	\label{eq_3}
	X^3 + A X + \alpha B=0,
\end{equation} 
where $X = \alpha x$, $A = a^2 + (\alpha  b + \alpha) a + b^2 +\alpha^2  b  + 1 $ and $B= a^3+\alpha b a^2 + (\alpha^2 b^2 + b^2 +1) a + \alpha^3b^3+\alpha b^2 + \alpha.$  

By Lemma \ref{lemma1}, $B\neq0$ holds under the case $(a,b)\neq (\alpha,1)$.  Let $h(t) = t^2+ \alpha  Bt+A^3$. By computation, we have 
$$\frac{A^3}{\alpha^2 B^2} = \frac{C}{\alpha B}+\frac{C^2}{\alpha^2 B^2}+ \frac{1}{\alpha^2},$$
where $$C = (b+1)(a^2+ \alpha (b+1)a + (\alpha^2+1)b^2 + \alpha^2(b+1)+1).$$
Thus $$\tr_m\left( \frac{A^3}{ \alpha^2 B^2}\right) = \tr_m \left( \frac{1}{\alpha^2}\right)  = \tr_m(1)$$ by Lemma \ref{lemma1}.  
Let $ \frac{1}{\alpha^2} = \beta + \beta^2+1$ with $\beta \in\fm$. Then also by Lemma \ref{lemma1}, we know that $\alpha (\beta + \omega) $ is not cubic, where  $\omega\in \mathbb{F}_{2^2}\backslash\f2$. Note that $$(\beta+\omega)(\beta+\omega^2) = \beta^2+\beta+1 =\frac{1}{\alpha^2} $$ and thus 
\begin{equation}
	\label{inverse-beta}(\beta+\omega)^{-1} = (\beta+\omega^2) \alpha^2.
\end{equation}
Moreover, the equation $h(t)=0$ has two solutions $t_1= C+ (\beta + \omega)\alpha  B$ and  $t_2= C+(\beta + \omega^2)\alpha B$ in $\fm$ (resp. $\mathbb{F}_{2^{2m}}$) if $m$ is even (resp. odd). We now show that $t_1$ and $t_2$ are not cubic. Let $\hat{a}=\frac{a}{\alpha}$. Then 
\begin{align*}
	C =& \alpha^2 (b+1) (\hat{a}^2+(b+1)\hat{a} + (1+\frac{1}{\alpha^{2}}) b^2 + (b+1) + \frac{1}{\alpha^{2}} )\\
 =	& \alpha^2 (b+1) (\hat{a}^2 + (b+1)\hat{a} + (\beta+\beta^2) b^2 + b+\beta + \beta^2)  \triangleq \alpha^2 C_1
\end{align*}
and 
\begin{align*}
	B =& \alpha^3 (\hat{a}^3 + b \hat{a}^2 + ( b^2 +  \frac{b^2 +1}{\alpha^2}) \hat{a} + b^3+ \frac{b^2 +1}{\alpha^2}) \\ 
 =	&\alpha^3 (\hat{a}^3 + b \hat{a}^2 + ( b^2 +  (b^2 +1) (\beta+\beta^2+1) ) \hat{a} + b^3+ (b^2 +1) (\beta+\beta^2+1)) \triangleq \alpha^3 B_1.
\end{align*}
Moreover, 
\begin{align*}
	t_1 =&  C+ (\beta + \omega)\alpha  B \\
	=& (\beta + \omega)\alpha^4 \left( (\beta+\omega)^{-1}\alpha^{-4} \alpha^2 C_1 + B_1 \right) \\
	=& (\beta + \omega)\alpha^4 \left( (\beta+\omega^2) C_1 + B_1  \right),
\end{align*}
where the last equality holds due to Eq.\ \eqref{inverse-beta}. Furthermore, by computing directly, we can find that 
$$(\beta+\omega^2) C_1 + B_1  = (\hat{a} + (\beta + \omega) b + \beta + \omega^2)^3.$$ Note that by Lemma \ref{lemma1}, $(\beta + \omega)\alpha^4$ is not cubic. Therefore, $t_1$ is not cubic. Similarly, $t_2$ is not cubic, either. 
 Then by Lemma \ref{cubic},  the equation $x^3 + Ax + \alpha B = 0$ has no solution in $\fm$. 

 	 Hence from Eq.\ \eqref{resultant}, we have $x = 0$ or $x=a$. Next, we  show that  Eqs.\ \eqref{eq_1} have only  $(x,y)\in\{ (0,0), (a,b) \}$ as solutions.  
 
 	If $a=0$, then $x=0$. Moreover, Eq.\ \eqref{eq_1_1} and  Eq.\ \eqref{eq_1_2} become $\alpha by^2+ \alpha b^2y=0$ and $\alpha by^4+ \alpha b^4y=0$, respectively. Thus $y=0$ or $y=b$. In the following, we assume that $a\neq0$. 
 
 {\bfseries Case 1:  $x=0$}. In this case  Eqs.\ \eqref{eq_1} become 
 	\begin{subequations} 
 		\renewcommand\theequation{\theparentequation.\arabic{equation}}  
 		\label{eq_1-x=0}  
 		\begin{empheq}[left={\empheqlbrace\,}]{align}
 			&  (\alpha b + a) y^2 + (\alpha b^2+a) y = 0 \label{eq_1-x=0_1} \\ 
 			&  (\alpha^2a + \alpha b +a ) y^4 + \alpha a^2 y^2 + (\alpha a^4+\alpha b^4+a) y  = 0.  \label{eq_1-x=0_2} 
 		\end{empheq}
 	\end{subequations} 
 	We now show that Eqs.\ \eqref{eq_1-x=0} have only one solution $y=0$ for any $(a,b)\in\fm\times \fm \backslash\{(0,0), (\alpha, 1)\}$. 
	 If $a = \alpha b \neq \alpha $, then $b \neq 0, 1$ and by Eq.\ \eqref{eq_1-x=0_1}, we get $\alpha (b^2+b) y =0$, i.e., $y=0$. If $a\neq \alpha b$ and $\alpha b^2+a = 0$, then by Eq.\ \eqref{eq_1-x=0_1},  $y=0$ clearly. Now we consider the subcase $(\alpha b + a)(\alpha b^2+a)\neq 0$.
In this subcase, from Eq.\ \eqref{eq_1-x=0_1}, we have  $y = \frac{\alpha b^2+a}{\alpha b + a}$. Plugging it into Eq.\ \eqref{eq_1-x=0_2} and simplifying, we obtain 
$$ (\alpha b^3 + a^3 + ab^2) ( a^3+\alpha b a^2 + (\alpha^2 b^2 + b^2 +1) a + \alpha^3b^3+\alpha b^2 + \alpha) = 0, $$  
which contradicts the condition that $ x^{3}+x+\alpha$ has no root in $\mathbb{F}_{2^m}$ and Lemma \ref{lemma1}. Therefore Eqs.\ \eqref{eq_1-x=0} have only one solution $y=0$.

	{\bfseries Case 2:} $x=a$. In this case  Eqs.\ \eqref{eq_1} become 
\begin{subequations} 
	\renewcommand\theequation{\theparentequation.\arabic{equation}}  
	\label{eq_1-x=a}  
	\begin{empheq}[left={\empheqlbrace\,}]{align}
		& (\alpha b + a) (y+b)^2 + (\alpha b^2+a) (y+b) = 0  \label{eq_1-x=a_1} \\ 
		& (\alpha^2a + \alpha b +a ) (y+b)^4 + \alpha a^2 (y+b)^2 + (\alpha a^4+\alpha b^4+a) (y+b)  = 0.  \label{eq_1-x=a_2} 
	\end{empheq}
\end{subequations} 
It is clear that $y=b$ is the unique solution of Eqs.\ \eqref{eq_1-x=a} by the discussions of the case when $x=0$. 

	To summarize, Eqs.\ \eqref{eq_1} {have only $(x,y) \in \{(0,0), (a,b)\}$ as solutions} in $\fm \times \fm$ for any $(a,b)\in \fm \times \fm\backslash\{(0,0)\}$. Therefore, $F$ is APN over $\fm\times\fm.$ 
\end{proof}}

As for the case of family $\cF_1$,  family F2 generalizes the one given in Theorem \ref{th:lzlq}. Indeed, our family includes this one and can be defined also over $\mathbb{F}_{2^{2m}}$, with $m\equiv 0\mod 3$.

\subsection{New APN functions from our families}

In this section we show that our constructions can produce new instances of APN functions. In particular, for $n=12$ ($m=6$), we obtain APN functions that are inequivalent to any APN function belonging to an already known family.

For this dimension, we consider the EA-invariant $N_F$ for all known APN functions, that is, functions belonging to a known family\footnote{Due to the hardness of checking the APNness condition for the  functions obtained via the generalized isotopic shift (GIS) construction, presented in \cite{seta18-1}, we do not consider these maps.
Moreover, if we restrict the coefficients of the polynomials to a subfield of $\mathbb{F}_{2^{12}}$, following the approach  done in \cite{seta18-1} for $n=9$, then all the functions obtained (for $n=12$) from the GIS construction are equivalent to a Gold map. It is not known whether this is true also for unrestricted coefficients.
}.

In Table \ref{tab:APNfamuni} and \ref{tab:APNfambiv}, we report all the known families of APN functions that can be defined over $\mathbb{F}_{2^{12}}$. {We divide the families based on whether they are presented  in univariate representation or in bivariate representation.} Moreover, among the monomial APN functions we report only the Gold APN function since, from the results in \cite{yoshi16}, we have that any quadratic APN function equivalent to a power function is equivalent to a Gold function.

\begin{table}[h!]
    \centering
\scriptsize{    \begin{tabular}{|c|c|c|c|}

    \hline
      $N^\circ$&Functions &  Conditions &  In\\
      \hline
      \hline
        1  & $x^{2^i+1}$& $\gcd(i,n)=1$& \cite{86} \\
          \hline
       & & $n=3k$,   gcd$(k,3)$= gcd($s,3k$)=1, & \\
		2& $x^{2^s+1}+u^{2^k-1}x^{2^{ik}+2^{mk+s}}$ &  $i=sk$ mod $3$, $m=3-i$,   & \cite{30} \\
		& & $n\geq12$, $u$ primitive in $\mathbb{F}_{2^n}^*$   & \\
     
        \hline
		& & $n=3k$,  gcd$(k,3)$= gcd($s,3k$)=1,   & \\
		3& $ux^{2^s+1}+u^{2^k}x^{2^{-k}+2^{k+s}}+$ &  $v,w \in \mathbb{F}_{2^k}$, $vw\neq1$,  & \cite{BraBy08,13} \\
		& $vx^{2^{-k}+1}+wu^{2^k+1}x^{2^s+2^{k+s}}$ & $3|(k+s)$ $u$ primitive in $\mathbb{F}_{2^n}^*$   & \\
		\hline
		&&$n=3m$, $\gcd(s,m)=1$, $v\in\mathbb{F}_{2^{m}}^{*}$&\\
	4	&$L(z)^{2^m+1}+vz^{2^m+1}$& $\mu\in\mathbb{F}_{2^{3m}}$ with $\mu^{2^{2m}+2^m+1}\neq1$&\cite{LZLQ}\\
	&&and $ L(z) = z^{2^{m+s}}+\mu z^{2^s}+z$ permutation &\\
		\hline
		   & & $q=2^m$, $n=2m$,  gcd($i,m$)=1,   & \\
	5& $sx^{{q+1}}+x^{2^i+1}+x^{q({2^i+1})}$ & $c\in\mathbb{F}_{2^n}$, $s\in\mathbb{F}_{2^n}\setminus\mathbb{F}_q$, & \cite{BudCar08}\\
		&$+ cx^{2^{i}q+1}+ c^{q}x^{{2^i+q}} $    &$z^{2^i+1}+ c z^{2^i} + c ^{q}z+1$& \\
		& & has no solution $x$ such that $x^{q+1}=1$      & \\
		\hline
		6& $x^3+a^{-1}\textit{Tr}(a^3x^9)$& $a\neq0$ & \cite{BCL09}\\
		\hline
		7& $x^3+a^{-1}\textit{Tr}_n^3(a^3x^9+a^6x^{18})$& $3|n$, $a\neq0$ & \cite{BCL09.2} \\
		\hline
		8& $x^3+a^{-1}\textit{Tr}_n^3(a^6x^{18}+a^{12}x^{36})$& $3|n$, $a\neq0$ & \cite{BCL09.2} \\
		
		\hline

    \end{tabular}
    \caption{Known APN families in univariate form over $\mathbb{F}_{2^{n}}$ that can be defined over $\mathbb{F}_{2^{12}}$}
    \label{tab:APNfamuni}}
\end{table}

\begin{table}[h!]
    \centering
    \scriptsize{
    \begin{tabular}{|c|c|c|c|}
    \hline
      $N^\circ$&Functions &  Conditions &  In\\
      \hline
      \hline
      &&$m$ even,&\\
      9	& $(xy,x^{2^k+1} +\alpha x^{(2^k+1)2^i})$ &  $\gcd(k,m)=1$,   & \cite{pott}\\
      &&$\alpha$ not a cube&\\
        \hline
        &&$\gcd(k,m)=1$,&\\
	10 & $(xy, x^{2^{3k}+2^{2k}}+ax^{2^{2k}}y^{2^k}+by^{2^k+1})$ & $z^{2^k+1}+az+b$  &   \cite{tan19}\\
		&&has no root in $\mathbb{F}_{2^{m}}$ &\\
	  \hline
		 &  & $m$ even, $\gcd(i,m)=1$,& \\
		11&$(xy, x^{2^i+1}+x^{2^{i+m/2}}y^{2^{m/2}}+bxy^{2^i}+cy^{2^i+1})$& $(cz^{2^i+1}+bz^{2^i}+1)^{2^{m/2}+1}+z^{2^{m/2}+1}$&\cite{CBC21}\\
		&& has no root in $\mathbb{F}_{2^{m}}$ &   \\
		  \hline
		 &  &$m\equiv 2 \pmod 4$, $\gcd(k,m)=1$ &\\
	12	&$(x^{2^i+1}+By^{2^i+1}, x^{2^{k+m/2}}y+\frac{a}{B}xy^{2^{k+m/2}})$&$B$ not a cube,&
		\cite{golkol}\\
		&& $a \in \mathbb{F}_{2^{m/2}}^*$, $B^{2^{k}+2^{k+m/2}} \ne a^{2^{k}+1}.$&\\
		\hline
	
    \end{tabular}
    \caption{Known APN families in bivariate forms over $\mathbb{F}_{2^{2m}}$ that can be defined over $\mathbb{F}_{2^{12}}$}
    \label{tab:APNfambiv}}
\end{table}

In Table \ref{tab1}, we report the values obtained for the invariant $N_F$ and the corresponding family.
\begin{table}[h!]
    \centering

    \begin{tabular}{|l|c|}
   \hline
     { \centering $ {N_F}$ }  & { $N^\circ$ in Table \ref{tab:APNfamuni} or \ref{tab:APNfambiv} }\\
      \hline
      \hline
             $[n_1=1365,n_2=100100,n_3= 99840,n_4=  91] $ & family n. 1 (Gold)\\
             \hline
       $[n_1=1365, n_2=100100,n_3=141700]$  & families n. 2, 3  \\
       $ [n_1=1365,n_2=100100,n_3=144010,n_4= 140]$ & and 4 \\
       \hline
       $[n_1=1365, n_2=100100,n_3= 142932 ] $ & families n. 5, 9, 10 and 11\\
       \hline
       $ [n_1=1365, n_2=100100,n_3= 139980,n_4=  152 ]$ & family n. 6 \\
       $ [n_1=1365,n_2= 100100,n_3= 138918,n_4=  103 ]$ & \\
       \hline
       $[n_1=1365,n_2= 100100, n_3=144728,n_4=  184 ] $ &  \\
       $[n_1=1365,n_2= 100100,n_3= 143502,n_4=  198 ] $ & families n. 7 and 8\\
       $[n_1=1365,n_2= 100100,n_3= 144584,n_4=  229 ] $ & \\
       $[n_1=1365,n_2= 100100,n_3= 143286,n_4=  155 ] $ & \\
       \hline
       $ [n_1=1365,n_2= 100100,n_3= 144759,n_4= 126 ]$ & family n. 12\\
       \hline
    $ [n_1=1365,n_2=100100,n_3=140664]$ & family F1 from Theorem \ref{th:main}\\
    \hline
    $ [n_1=1365,n_2=100100,n_3=144198,n_ 4=192]$ &family F2 from Theorem \ref{th:main2}\\
    \hline
    \end{tabular}
    \caption{EA-invariant $N_F$ for the known APN families in $\mathbb{F}_{2^{12}}$}
    \label{tab1}
\end{table}

Due to the result of Yoshiara \cite{ccz-ea}, CCZ-equivalence between quadratic APN functions coincides with EA-equivalence. 
Therefore, the invariant $N_F$ tells us that our construction is not included in any known family. {Hence, we obtained two new families of APN functions.}

\section{Concluding remarks}
We have introduced two new families of APN quadratic functions using biprojective polynomials and the Dillon's method. These constructions include two other families, family $\mathcal{F}_1$ in \cite{GOLOGLU} and the first family in \cite{LZLQ}. Our constructions allow to obtain APN functions over $\mathbb{F}_{2^{2m}}$ also for the case $3\mid m$.

In \cite{golkol}, the authors have investigated the equivalence between APN functions obtained from the biprojective framework. In particular, from  family $\cF_1$ we can get $\varphi(m)/2$ inequivalent functions ($\varphi$ is the Euler's function), one for each $k\le m/2$ coprime with $m$. It would be interesting to determine if, for functions defined in Theorem \ref{th:main}, for a fixed $k$ we can get more than one function (up to equivalence) using different values of $\alpha$.
The same for functions derived from Theorem \ref{th:main2}.

Another interesting problem is determining the Walsh spectra of our functions, that is the set $\{W_F(a,b)\,:\,a,b\in\mathbb{F}_{2^n}\}$. 
For the case of family $\cF_1$, we have a Gold-like Walsh spectrum \cite{KKK}. This has been determined showing that the functions in $\cF_1$ are 3-to-1.
Indeed, for quadratic APN functions we have the following property.
\begin{theorem}[\cite{KKK}]\label{th:3to1}
Let $n$ even and $F$ be a quadratic APN function over $\mathbb{F}_{2^n}$.
If $F(0)=0$ and for any $b\ne0$ in the image of $F$ we have at least three pre-images, then
$F$ is 3-to-1, that is, we have $|F^{-1}(0)|=1$ and for all other $b\in\mathrm{Im}(F)\setminus\{0\}$ we have $|F^{-1}(b)|=3$. Moreover, $F$ has a Gold-like Walsh spectrum, i.e. $\{|W_F(a,b)|\,:\,a,b \in \mathbb{F}_{2^n}\}=\{0,2^{n/2},2^{{n/2}+1}\}$.
\end{theorem}

For the case $k=1$, it is possible to show that APN functions from Theorem \ref{th:main} are 3-to-1, and thus we have a Gold-like Walsh spectrum.

\begin{theorem}
Let $k=1$ and $m>0$. Let $\alpha\in\mathbb{F}_{2^m}$ be such that $x^3+x+\alpha$ has no root in  $\mathbb{F}_{2^m}$. Then, the APN function given in Theorem \ref{th:main}
$$
F(x,y)=(x^3+xy^2+\alpha y^3,x^5+\alpha x^4y+(1+\alpha^2)xy^4+\alpha y^5)
$$ 
has Gold-like Walsh spectrum.
\end{theorem}
\begin{proof}
From Lemma \ref{lemma1} we have that $\tr_m(1/\alpha^2)=1$ and there exists $\beta$ such that $\beta^2+\beta+1=1/\alpha^2$.

Then, let us consider the linear transformation $L(x,y)=((\beta+1)x+y/\alpha,x/\alpha+\beta y)$. $L$ is a bijection and it is easy to check that 
for any $(x,y)\in\mathbb{F}_{2^m}\times\mathbb{F}_{2^m}$ we have
$$
L^2(x,y)=L(x,y)+(x,y), \text{ and }L^3(x,y)=(x,y).
$$

By a direct computation, using the fact that $\beta^2+\beta+1=1/\alpha^2$, we have that $F(L(x,y))=F(x,y)$ for any $(x,y)\in\mathbb{F}_{2^m}\times\mathbb{F}_{2^m}$. 

Therefore, for any $b\ne0$ in the image of $F$ we have at least three pre-images, and thus, from Theorem \ref{th:3to1} we have that $F$ has  Gold-like Walsh spectrum.
\end{proof}

From computational results, this seems to be true also for the general case of Theorem \ref{th:main}. 

\section{Acknowledgment}

 {The work of Kangquan Li is supported by the National Natural Science Foundation of China under Grant (No.~62172427).}

\end{document}